\documentclass[11pt,a4paper]{article}

\usepackage{fullpage}

\usepackage{amsmath,amssymb,amsthm,amsopn}
\usepackage{authblk}

\usepackage{algorithm}
\usepackage[noend]{algpseudocode}

\usepackage{color}
\long\def\red#1\par{{\color{red}#1}}
\long\def\ble#1\par{{\color{blue}#1}}

\usepackage{enumitem}

\def\ie{{\it i.e.,}~}
\def\eg{{\it e.g.,}~}

\def\etal{{\it et al.}~}

\renewcommand\Pr{\mathbb{P}}

\definecolor{orange}{rgb}{1,0.5,0}

\def\dynamicqs{\textsc{Quicksort}}

\def\hrow{@{\hspace{5pt}}}

\def\tpi{\tilde\pi}

\DeclareMathOperator*{\E}{\mathbb{E}}
\DeclareMathOperator*{\pr}{\mathbb{P}}

\algnewcommand{\algorithmiccall}{\textbf{call}~}
\algnewcommand{\algorithmicto}{\textbf{to}~}

\newtheorem{lemma}{Lemma}[section]

\newtheorem{theorem}{Theorem}[section]

\newtheorem{claim}{Claim}[section]
\newtheorem{definition}{Definition}[section]
\theoremstyle{definition}

\newtheorem{remark}{Remark}[section]

\begin{document}

\author[1]{Varun Kanade\thanks{\texttt{varunk@cs.ox.ac.uk}}}
\author[2]{Nikos Leonardos\thanks{\texttt{nikos.leonardos@gmail.com}}}
\author[3]{Fr\'ed\'eric Magniez\thanks{\texttt{frederic.magniez@cnrs.fr}}}
\affil[1]{University of Oxford}
\affil[2]{University of Athens}
\affil[3]{CNRS and IRIF, Univ Paris Diderot, Sorbonne Paris-Cit\'e, France}
\title{Stable Matching with Evolving Preferences\thanks{%
		Partially supported by the French ANR Blanc project ANR-12-BS02-005
		(RDAM) and the European Commission IST STREP project Quantum Algorithms
		(QALGO) 600700.  Most of the work was carried out when V.K and N. L. were
		at LIAFA (now IRIF).  During this time V.K. was supported by the
		Fondation Sciences Mathématiques de Paris (FSMP).  }}
\date{}
\maketitle

\begin{abstract}
  We consider the problem of stable matching with dynamic preference lists. At
  each time-step, the preference list of some player may change by swapping
  random adjacent members. The goal of a central agency (algorithm) is to
  maintain an approximately stable matching, in terms of number of blocking
  pairs, at all time-steps. The changes in the preference lists are not
  reported to the algorithm, but must instead be probed explicitly. We design
  an algorithm that in expectation and with high probability maintains a
  matching that has at most $O((\log n)^2)$ blocking pairs.
\end{abstract}

\section{Introduction}

In the world of massive and distributed data, it is hardly reasonable to assume
that data are static. Yet, one must design algorithms that maintain a solution
for a given problem that is (approximately) consistent with the requirements,
\eg a permutation that is almost sorted. Thus, it is important to design
algorithms and data structures that are robust to changes in their input, \ie
they produce an output with some performance guarantee (quickly).

There are a few different dynamic data models that have been considered. The
area of dynamic graph algorithms consists of maintaining some property or
structure, such as connectivity, matchings, or spanning trees, even when the
underlying graphs are changing~\cite{EGI:1999,OR:2010,NS:2013,GP:2013}. Here, it
is assumed that the changes to the graph may be \emph{arbitrary}, but are
reported to the algorithm; and the focus is on designing data structures and
algorithms that adapt efficiently (typically in terms of computational time) to changes in
the input. The area of streaming algorithms studies the setting where the data
can only be accessed as a \emph{stream} and the focus is on producing the
desired output with highly space-efficient procedures (typically
poly-logarithmic in the size of the input). In the area of online algorithms,
one must design procedures that, even when data is revealed bit by bit, produce
an output that is \emph{competitive} with algorithms that see the entire input
at once.

Recently, Anagnostopoulos \etal~\cite{akmu2011} proposed the \emph{evolving
data model} to take into account the dynamic aspects of massive data. In this
model, the changes to the data are not revealed to the algorithm; instead, an
algorithm has query access to the data.  However, it is assumed that the
changes to the data are \emph{stochastic}, not adversarial. In this setting,
the focus is not on computational complexity (which is allowed to be polynomial
at each time-step), but query complexity, the number of probes made by the
algorithm. Anagnostopoulos \etal~\cite{akmu2011} studied the problem of
maintaining the \emph{minimum} element of a permutation and an approximately
sorted permutation, motivated by questions such as maintaining high
(page)-ranked pages. In their setup, a permutation evolves by choosing a random
element and swapping it with an adjacent element. In later work,
Anagnostopoulos~\etal~\cite{AKMUV12} studied evolving graph models and problems
such as $s$-$t$ connectivity and minimum spanning tree. 

In this work, our aim is to bring this notion to game theory starting from the
basic problem of computing a stable matching. In other words, we introduce the
notion of evolving agents, who may not report any updates to their strategy (or
preferences) without an explicit request. In the stable matching problem,  the
input consists of two sets $A, B$ of equal size, and for each member a total
order (preference) over members in the other set. Given a matching between $A$
and $B$, a pair $(x, y)$ with $x \in A$ and $y \in B$ is blocking if they prefer
each other to their matches. A matching is stable if there are no blocking
pairs. Gale and Shapley showed that a stable matching always exists and can be
found by an efficient algorithm~\cite{GS:1962}. We consider the setting where
the preference lists \emph{evolve} over time. The preference lists can evolve
over time, by swapping adjacent elements. More precisely, while the algorithm
can perform one query per time-step, we assume that a total number of $\alpha$
swaps events also occur, where $\alpha=\Theta(1)$ is some fixed parameter,
called the \emph{evolution rate}.  This assumption is similar to previous works
and models the critical regime: with less evolution events the input is
basically static, and with more the input evolves too fast. The goal is to
maintain a matching that has \emph{few} blocking pairs.

We summarize our results as follows. All three statements hold in expectation
and with high probability.

\begin{enumerate}[label=(\roman*)]
	\item Using the results of Anagnostopoulos \etal~\cite{akmu2011} for sorting
		permutations, we design an algorithm that maintains a matching with at most
		$O(n\log n)$ blocking pairs, at all time-steps after roughly the first $n^2
		\log n$ steps (Theorem~\ref{thm.matchsimple}). Also, we observe that any analysis 
		that uses their method as a black-box, cannot improve on this bound
		(Remark~\ref{remark}).\footnote{We don't rule out the possibility that a
			more fine-grained analysis of the algorithm may give better bounds;
			instead we design new algorithms.}
	\item In a restricted setting, where only one side, say the $B$ side, has
		evolving preference lists, and if the $A$ side has uniform random
		permutations as preference lists (known to the central agency), we design an
		algorithm that maintains a matching with $O(\log n)$ blocking pairs at all
		time-steps after roughly the first $n \log n$ steps
		(Theorem~\ref{thm.onesided}).
	\item Finally, we design an algorithm in the general setting, that maintains
		a matching with at most $O((\log n)^2)$ blocking pairs at all time-steps after roughly
		the first $n^2\log n$ steps (Theorem~\ref{thm.improved}).
\end{enumerate}

\section{Preliminaries}

In the rest of the paper, $n\geq 1$ denotes an integer parameter and $[n]$ the
set of integers $\{1,2,\ldots,n\}$.  For a non-negative random variable $X$,
parametrized by some integer $n$, we write ``$X=O(f(n))$ \emph{in expectation
and with high probability}'' when for any constant $c$ there exist constants
$n_0,c',c''>0$ such that $\E[X]\leq c'f(n)$ and $\Pr(X>c''f(n))\leq n^{-c}$, for
every integer $n\ge n_0$. 
	
\subsection{Stable Matching} 

We only consider the bipartite stable matching problem, also known as stable
marriage. There are two sets of players $A$ and $B$, with $|A|=|B|=n$. Each
player $x\in A$ ($y\in B$) holds a permutation of $B$ ($A$), denoted
$\pi_x:B\to[n]$ ($\pi_y:A\to[n]$) indicating their preferences over players in
the set $B$ ($A$). Thus, for $y \in B$, $\pi_x(y)$ denotes the rank of $y$ in
$x$'s preference list (where $1$ is the highest rank).

Let $M : A \rightarrow B$ be a matching (a bijection). 
A pair $(x, y)$ is said to be \emph{blocking} if $y \prec_{\pi_x} M(x)$ and
$x \prec_{\pi_y} M^{-1}(y)$, where $z\prec_{\pi}z'$ indicates that $z$ is
ranked higher than $z'$ according to permutation $\pi$ (\ie $\pi(z)<\pi(z')$). 
In words, $x$ prefers $y$ to $M(x)$ and $y$ prefers $x$ to $M^{-1}(y)$. 

A matching $M$ is \emph{stable} if there are no blocking pairs.  Then the
\emph{stable matching problem} is to find a stable matching given preference
lists $\{ \pi_z:z\in A\cup B)\}$.  Gale and Shapley~\cite{GS:1962} proved that a
stable matching always exists, and gave an algorithm that given the preferences
lists as input finds a stable matching in $O(n^2)$ time.

The Gale-Shapley algorithm is simple to describe. Only players in the set $A$
make proposals. Initially all players are \emph{unmatched}. Let $M$ denote a
partial matching at some point. If there is an unmatched player $x
\in A$, $x$ makes a proposal to $y \in B$, where $y$ is the highest-ranked
player in $\pi_x$ to whom $x$ has not yet proposed.  If $y$ is unmatched, or
prefers $x$ to $M^{-1}(y)$, then $y$ accepts the proposal and we set $M(x) = y$.
In the latter case, the agent previously matched to $y$, \ie $M^{-1}(y)$ before
$M$ was updated, becomes unmatched once more.  Gale and Shapley showed that this
algorithm always results in a stable matching. 

Wilson~\cite{Wilson:1972} studied the problem
where all the preference lists are independent and uniformly random
permutations; in this case, he showed that the number of proposals made by the
Gale-Shapley algorithm is $O(n \log n)$ in expectation and with high
probability (see also \cite{mr1995}). In fact, only the proposing side needs
to be random in their statement. We provide a proof sketch for completeness.

\begin{theorem}[\cite{Wilson:1972}]\label{thm:proposalbound}
  If the permutations $\{\pi_x:x \in A\}$ are chosen randomly, the number
  of proposals made in the Gale-Shapley algorithm (where only $A$ makes
  proposals) is $O(n\log n)$ in expectation and with high probability.
\end{theorem}
\begin{proof}[Proof sketch]
  Following the proof in \cite{mr1995} (see also \cite{knuth1997}), analyze an
  alternative procedure where every proposal is uniform over the whole of $B$.
  If it happens that $x\in A$ proposes to a $y\in B$ that has already rejected
  $x$, then a rejection is guaranteed.  It is not hard to show that the number
  of proposals such an algorithm makes stochastically dominates the number of
  proposals of the classical algorithm.  Next, by the method of deferred
  decisions, fix the randomness in the algorithm only when needed. Then observe
  that the number of proposals is equal to the number of coupons chosen in the
  coupon-collector's problem.
\end{proof}

\subsection{Model for evolving input}

A general framework for studying dynamic data was introduced in
\cite{akmu2011}. Here we are only concerned with evolving preference lists (or
permutations).  In our model, we consider discrete \emph{time-steps}. In each
time-step, the algorithm is competing against \emph{nature} as follows: the
algorithm can \emph{query} the input locally, \emph{nature} lets the input
evolve according to one or more \emph{evolution events}.

A query to the stable matching input is a triplet $(z,u,v)\in (A\times B^2)\cup
(B\times A^2)$ and the answer is whether $\pi_z(u)<\pi_z(v)$. One evolution
event consists of the following: pick $z\in A\cup B$ and $i\in[n-1]$ uniformly
at random and swap $u=\pi^{-1}_z(i)$ and $v=\pi^{-1}_z(i+1)$ (\ie set
$\pi_z(u)=i+1$ and $\pi_z(v)=i$).

While the algorithm can perform one query per time-step, $\alpha$ evolution
events also occur, where $\alpha\geq 1$ is some integer called the
\emph{evolution rate}.  We further assume that $\alpha=\Theta(1)$, meaning that
evolution events occur basically as often than the algorithm probes. We
emphasize that the rate-limiting factor is the queries made by the algorithm.
In particular, the algorithm may perform arbitrary (polynomial-time)
computations in between time-steps. We are now ready to define our problem: 
\begin{quote}
\textbf{Evolving Stable Matching} (ESM): Given query access to an instance of
the stable matching problem with evolution rate $\alpha=\Theta(1)$, maintain a
matching that minimizes the number of blocking pairs.
\end{quote}

\subsection{Sorting evolving permutations}

The problem of sorting a single evolving permutation has been already addressed
in~\cite{akmu2011}.  In this context, the evolution rate is still constant, but
denotes the evolution rate of this single permutation.  We will use the
algorithm \dynamicqs{} of~\cite{akmu2011}.  It is simply the randomized version
of quicksort which is shown to be robust with respect to an evolving input. The
first lemma shows that the running time of quicksort is not affected by
evolution events.\footnote{We remark that Anagnostopoulos \etal~\cite{akmu2011}
use `whp' to denote events that hold with probability $1 - o(1)$, rather than the
stronger notion we use in this paper. However, their proofs for the results
used in our paper actually prove the stronger bounds. They have other results
that do not satisfy the stronger notion and these are not used in our work.}

\begin{lemma}[Proposition~3 in \cite{akmu2011}]\label{lem.quicksort}
  The running time of \dynamicqs{} is $O(n\log n)$ in expectation and with
  high probability, for any rate of evolution when the pairs to be swapped are
	chosen randomly.
\end{lemma}

Second, Lemma~6 in \cite{akmu2011} states that \dynamicqs\ when run on an
evolving permutation $\pi$, computes a permutation $\tpi$ in which every element
is approximately sorted.  At time-step $t$, let $\pi^t$ the denote the current
permutation, and $\tpi^t$ its approximation computed by the algorithm.

\begin{lemma}[Lemma~6 in \cite{akmu2011}]\label{lem.six}
	Let $t$ be the time-step of completion of \dynamicqs, then given an element $u$,
	the number of pairs $(u,v)$ that the permutations $\pi^t$ and $\tpi^t$ rank
	differently is $O(\log n)$ in expectation and with high probability.
\end{lemma}

In our setting there are $2n$ evolving permutations over some set of $n$ elements.
Algorithm~\ref{alg.sequential} simply sorts $m$ (out of $2n$) permutations,
denoted by $\pi_1, \ldots, \pi_m$ using \dynamicqs\ one after another. (We
always invoke Algorithm~\ref{alg.sequential} with either $n$ or $2n$
permutations.)

\begin{algorithm}
  \caption{\label{alg.sequential}:~~Sequential sorting\strut}   
	\begin{algorithmic}[1]
		\Procedure{SequentialSort}{$\{\pi_j ~:~ j = 1, \ldots, m\}$} \Comment{Only
		have query access to input}
		\For{$j=1$ \algorithmicto $m$}
      \State $\tpi_j\gets\dynamicqs(\pi_j)$
    \EndFor
		\State \Return $\{\tilde{\pi}_j ~:~ j = 1, \ldots, m \}$
		\EndProcedure
  \end{algorithmic}
\end{algorithm}

Using Lemma~\ref{lem.six} (Lemma~6 of \cite{akmu2011}) we can argue that
Algorithm~\ref{alg.sequential} maintains all permutations approximately sorted.
While the evolving rate is still $\alpha=\Theta(1)$, there are now $2n$ evolving
permutations, and the total number of evolution events is $\alpha$ per time-step.

\begin{lemma}\label{lem.logn}
	Let $t$ be the time-step when Algorithm~\ref{alg.sequential} terminates. Then, for
	$m \leq 2n$, given any element $u$ and $j\in[m]$, the number of pairs $(u,v)$
	that the permutations $\pi_j^t$ and $\tpi_j^t$ rank differently is $O(\log n)$
	in expectation and with high probability.
\end{lemma}
\begin{proof}
	Fix some $j\in[n]$. Suppose that $\tpi_j^t$ was computed at time-step $t'\le
	t$ (the time-step when \dynamicqs\ for this particular list terminates). By
	Lemma~\ref{lem.six} the statement holds for $u$ at time-step $t'$. Due to
	Lemma~\ref{lem.quicksort} we have $t-t'=O(n^2\log n)$ with high probability.
	During this time, the number of evolution steps that have swapped $u$ with
	an adjacent element is $O(\log n)$ with high probability. This follows from
	a balls-and-bins experiment where we throw $O(n^2\log n)$ balls
	(corresponding to the evolution steps) into $m (n-1)$ bins (corresponding to
	the adjacent pairs). It is known (see Exercise 3.1 in \cite{mr1995}) that in
	this particular case the number of balls in every bin is of the order of its
	mean with high probability. Therefore, during this time, at most $O(\log n)$
	more elements may be swapped with $u$.
\end{proof}

\subsection{Chernoff Bound with dependent variables}
We will require the following extension of the Chernoff bound. 
It follows from (the more general) Theorem~3.8 in~\cite{mcdiarmid1998}.

\begin{theorem}\label{thm:extended-chernoff}
For $i\in[n]$, let $Y_i$ be a random variable over some set $\mathcal{Y}_i$ and
$X_i$ be a Boolean random variable.
For any $y\in\prod_{i=1}^n\mathcal{Y}_i$ and $k\in[n]$, let $E_k(y)$ denote the event
$Y_1=y_1,\dots,Y_k=y_k$. 
Suppose $\pr[X_k=1|E_{k-1}(y)]\le p$, for all $y$ and $k$ as above.
Then, for any $t\ge0$,
\[\pr\Bigl[\sum X_k\ge pn+t\Bigr]\le\exp\Bigl({-\frac{3t^2}{6pn+2t}}\Bigr).\]
\end{theorem}


\section{Two simple cases}

In this section we present two simple arguments in two different settings.
First, we consider how the original Gale-Shapley algorithm performs when run,
without any modification, on lists produced by running quicksort on the evolving
input. We present a simple analysis showing a bound of $O(n\log n)$ on the
number of blocking pairs, which is better than the trivial bound of order $n^2$.
Next, we analyze the Gale-Shapley algorithm when evolution events only occur on
one side and the preference lists are uniformly random permutations; in
particular the preference lists on the $A$ side are chosen uniformly at random,
and the preference lists on the $B$-side are subject to evolution
events.\footnote{
	Note that after sufficiently many time-steps (though still polynomial) the evolution
	events ensure that all permutations are uniformly random. This follows from
	analyzing the mixing time of the corresponding Markov chain over permutations.
	See for example the book~\cite{LPW:2009}.} 
We present a simple analysis showing an $O(\log n)$ bound on the number of
blocking pairs for this special case.

\subsection{A simple algorithm}

Our first algorithm ignores evolution of preference lists and runs the standard
Gale-Shapley algorithm to produce a matching. More specifically, it first
obtains the preferences lists for all $2n$ agents using the \dynamicqs\
algorithm of~\cite{akmu2011} (\ie using Algorithm~\ref{alg.sequential}) and then
produces a matching by running the Gale-Shapley algorithm on these lists
(ignoring the fact that these lists are only \emph{approximately} correct). 

We show that this simple algorithm maintains a matching with at most $O(n\log
n)$ blocking pairs. Note that the number of blocking pairs is trivially at most
$n^2$. We further argue that improving the bound would require new ideas that
either go around Lemma~\ref{lem.logn} (Lemma~6 of \cite{akmu2011}) or improve
the analysis in a substantial way.

Algorithm~\ref{alg.matchsimple} runs in perpetuity. The matching $M$ is
maintained as the output until the new matching based on the newly sorted
preference lists can be computed.

\begin{algorithm}
  \caption{\label{alg.matchsimple}:~~Simple dynamic stable matching\strut}
  \begin{algorithmic}[1]
    \While{\textsc{True}}
		\State $\{\tilde{\pi}_z~:~z \in A \cup B\} \leftarrow
		\textsc{SequentialSort}(\{\pi_z~:~z \in A \cup B \})$
		\Comment{Calling Algorithm~\ref{alg.sequential}}
      \State\Return  Gale-Shapley matching $M$ on the (approximately) sorted lists
        $\{\tilde\pi_z : z \in A\cup B\}$
    \EndWhile
  \end{algorithmic}
\end{algorithm}

\begin{theorem}\label{thm.matchsimple}
	For a sufficiently large constant $c_0$ and any time-step $t \geq c_0 n^2
	\log n$, Algorithm~\ref{alg.matchsimple} maintains a matching with $O(n\log
	n)$ blocking pairs in expectation and with high probability. \\
\end{theorem}
\begin{proof}
	We consider the number of blocking pairs at time-step $T\le t$ when the current
	matching $M$ was computed. In the following discussion, we use the following
	notation: for $x \in A, y \in B$, if $M(x) = y$, then $M(y) = x$ (rather than
	$M^{-1}$). At time-step $T$, for each $z\in A\cup B$, define the indicator function
	$I_z(w)$ to be 1 when $M(z)\prec_{\tilde{\pi}_z} w$ and $w \prec_{\pi_z} M(z)$
	and 0 otherwise. (We don't explicitly use superscripts on the preference lists
	$\pi$ as time-step $T$ is fixed until specified otherwise.) By
	Lemma~\ref{lem.logn}, in expectation and with high probability, 
  \begin{equation}\label{eq.1}
		\sum_wI_z(w)=\bigl|\{w:w \prec_{\pi_z} M(z)\hbox{ and } M(z)
		\prec_{\tilde{\pi}_z} w\}\bigr|=O(\log n),
  \end{equation}
	If a pair $(x,y)$ is blocking at time-step $T$, then $x \prec_{\pi_y}M(y)$ and $y
	\prec_{\pi_x} M(x)$. Assume $x\prec_{\tilde{\pi}_y} M(y)$ and $y
	\prec_{\tilde{\pi}_x} M(x)$. Since $y \prec_{\tilde{\pi}_x} M(x)$, $x$ must
	have proposed to $y$ at some point during the execution of the Gale-Shapley
	algorithm. By the properties of the Gale-Shapley algorithm, $y$ should have
	been matched to an element of $A$ with rank according to $\tilde{\pi}_y$ at
	least as high as the rank of $x$ in $\tilde{\pi}_y$. This contradicts $x
	\prec_{\tilde{\pi}_y} M(y)$. It follows that either $M(y)\prec_{\tilde{\pi}_y}
	x$ or $M(x)\prec_{\tilde{\pi}_x} y$.  Define $U(x,y)$ to be 1 when $(x,y)$ is
	blocking and 0 otherwise.  We have argued that
  \[U(x,y)\le I_x(y)+I_y(x).\]
  By the union bound, Equation~\ref{eq.1} holds for every $z\in A\cup B$
  with high probability. Summing over all pairs $(x,y)$ and applying the
  union bound again
  \[\sum_{x,y}U(x,y)\le\sum_{x,y}I_x(y)+I_y(x)
      =\sum_x\Bigl(\sum_yI_x(y)\Bigr)+\sum_y\Bigl(\sum_xI_y(x)\Bigr)
      =O(n \log n),\]
  in expectation and with high probability.

  Next, let $t^\prime = t - T$. We need to account for blocking pairs that may
  have arisen during these $t^\prime$ time-steps. First, we observe that by
  Lemma~\ref{lem.quicksort} and union bound, with high probability $t^\prime =
  O(n^2 \log n)$ (as otherwise another matching more up-to-date than the one at
  time-step $T$ would be available). During the $t^\prime$ time-steps from $T$
  to $t$, evolution may create a blocking pair only if the swap decreases the
  rank of $M(z)$ in $\pi_z$, for some $z\in A\cup B$. Therefore, each step of
  the evolution introduces---independently---a new blocking pair with
  probability at most $\alpha/(n-1)$. The expected number of blocking pairs
  introduced is therefore at most $O(n\log n)$ for $\alpha = \Theta(1)$ (and
  assuming $t^\prime = O(n^2 \log n)$). The result follows by a simple
  application of the Chernoff and union bounds.
\end{proof}

\begin{remark}\label{remark}
We present a pair of instances to the stable matching problem with preference
lists $\pi_z$ and $\tilde\pi_z$ for $z \in A \cup B$, for which the conclusion
of Lemma~\ref{lem.logn} is satisfied, \ie for any $z$ the number of pairs $(i,
j)$ that are ordered differently in $\pi_z$ and $\tilde\pi_z$ are $O(\log n)$.
We then show that a matching that is stable with respect to $\tilde\pi$ as
preference lists has $\Omega(n \log n)$ blocking pairs with respect to $\pi$.
Thus, it follows that using the \dynamicqs\ algorithm of Anagnostopoulos
\etal\cite{akmu2011} and its analysis as a blackbox will not result in a
stronger result than the one provided in Theorem~\ref{thm.matchsimple}. 

First, we define the preference lists $\tilde{\pi}_z$ for $z \in A \cup B$. Let
$A = B = [n]$. Then for $x \in A$, the preference list (ranking) $\tilde\pi_x$
is defined as $x, x + 1, \ldots, n, 1, 2, \ldots, x - 1$. On the other hand for
$y \in B$, the preference list (ranking) $\tilde\pi_y$ is defined as $y, y - 1,
\ldots, 1, n, n - 1, \ldots, y + 1$. The rankings $\pi_z$, $z \in A \cup B$, are
now defined as follows: let $k$ be some parameter, $\pi_z$ simply as the
elements at rank $1$ and $k$ swapped.  Figure~\ref{fig:example} shows an example
with $n = 7$ and $k = 3$.  Clearly, when $k = \Theta(\log n)$, $\pi_z$ and
$\tilde\pi_z$ satisfy the conclusion of Lemma~\ref{lem.logn}. Yet, it is easy to
see that $M(i) = i$ is a stable matching for the preference lists $\tilde\pi_z$,
$z \in A \cup B$, and for this matching with respect to the preference lists
$\pi_z$, $z \in A \cup B$, every pair $(i, j)$ with $0 < j - i < k$ is a
blocking pair.
\end{remark} 

\begin{figure}
\[
\begin{array}{|c||c\hrow c\hrow c\hrow c\hrow c\hrow c\hrow c|}
	\multicolumn{8}{c}{\mbox{$A$-side lists } \pi}\\\hline
  1&3&2&1&4&5&6&7\\[-3pt]
  2&4&3&2&5&6&7&1\\[-3pt]
  3&5&4&3&6&7&1&2\\[-3pt]
  4&6&5&4&7&1&2&3\\[-3pt]
  5&7&6&5&1&2&3&4\\[-3pt]
  6&1&7&6&2&3&4&5\\[-3pt]
  7&2&1&7&3&4&5&6\\
\hline\end{array}
\quad
\begin{array}{|c||c\hrow c\hrow c\hrow c\hrow c\hrow c\hrow c|}
\multicolumn{8}{c}{\mbox{$B$-side lists } \pi}\\\hline
  1&6&7&1&5&4&3&2\\[-3pt]
  2&7&1&2&6&5&4&3\\[-3pt]
  3&1&2&3&7&6&5&4\\[-3pt]
  4&2&3&4&1&7&6&5\\[-3pt]
  5&3&4&5&2&1&7&6\\[-3pt]
  6&4&5&6&3&2&1&7\\[-3pt]
  7&5&6&7&4&3&2&1\\
\hline\end{array}
\quad
\begin{array}{|c||c\hrow c\hrow c\hrow c\hrow c\hrow c\hrow c|}
  \multicolumn{8}{c}{\mbox{$A$-side lists } \tilde\pi}\\\hline
  1&1&2&3&4&5&6&7\\[-3pt]
  2&2&3&4&5&6&7&1\\[-3pt]
  3&3&4&5&6&7&1&2\\[-3pt]
  4&4&5&6&7&1&2&3\\[-3pt]
  5&5&6&7&1&2&3&4\\[-3pt]
  6&6&7&1&2&3&4&5\\[-3pt]
  7&7&1&2&3&4&5&6\\
\hline\end{array}
\quad
\begin{array}{|c||c\hrow c\hrow c\hrow c\hrow c\hrow c\hrow c|}
\multicolumn{8}{c}{\mbox{$B$-side lists } \tilde\pi}\\\hline
  1&1&7&6&5&4&3&2\\[-3pt]
  2&2&1&7&6&5&4&3\\[-3pt]
  3&3&2&1&7&6&5&4\\[-3pt]
  4&4&3&2&1&7&6&5\\[-3pt]
  5&5&4&3&2&1&7&6\\[-3pt]
  6&6&5&4&3&2&1&7\\[-3pt]
  7&7&6&5&4&3&2&1\\
\hline\end{array}
\]
\caption{\label{fig:example} Instances $\pi$ and $\tilde\pi$ demonstrating
tightness of Algorithm~\ref{alg.matchsimple}.}
\end{figure}

\subsection{One-sided evolution}

In this section we analyze how the Gale-Shapley algorithm performs when the
initial preference lists are random, but there is no evolution on the lists of
the elements in $A$. Furthermore, we are going to assume that the algorithm
knows each permutation in $\{\pi_x:x\in A\}$. We call this setting
\emph{one-sided evolution.}

\begin{algorithm}
	\begin{algorithmic}[1]
		\caption{Gale-Shapley Algorithm}\label{alg:GS}
		\State $M \leftarrow \phi$ \Comment{Initialize empty matching}
		\For{$x \in A$} 
 			\State \texttt{new\_match} $\leftarrow$ \textsc{False}
			\State $p \leftarrow x$
			\While{\texttt{new\_match} = \textsc{False}}
				\State $y \leftarrow$ \texttt{first as yet unproposed as per }$\pi_p$
				\If{$M(y)$ \texttt{not yet set}} 
					\State $M(p) \leftarrow y$
					\State $M(y) \leftarrow p$
					\State \texttt{new\_match} = \textsc{True}
					\ElsIf{$p \prec_{\pi_y} M(y)$} \Comment{$y$ prefers $p$ to $M(y)$}
					\State $p^\prime \leftarrow M(y)$
					\State $M(y) \leftarrow p$
					\State $M(p) \leftarrow y$
					\State $p \leftarrow p^\prime$
				\EndIf
			\EndWhile
		\EndFor  
		\State\Return $M$
	\end{algorithmic}
\end{algorithm}

In this setting, the standard Gale-Shapley algorithm is implemented (basic
pseudocode is shown in Algorithm~\ref{alg:GS}). Note that the only time the
preference lists on the $B$-side are used is in Line 11. Thus, it is only for
these steps that we need to query the input (since the preference lists on the
$A$-side are known to the algorithm). Thus, the number of queries made by the
algorithm is bounded by the number of proposals. It was already observed that
the number of proposals made in a random instance of the stable matching problem
is $O(n \log n)$. The actual algorithm keeps implementing the Gale-Shapley
algorithm from scratch after completion. The matching from the previous
completed run is used as the current matching. We prove the following result for
the one-sided evolution setting.

\begin{theorem}\label{thm.onesided}
	For a sufficiently large constant $c_0$ and any time-step $t \geq c_0 n \log
	n$, the Gale-Shapley algorithm (repeatedly run and using the matching of the
	last completed run as the output) under one-sided evolution maintains a
	matching with at most $O(\log n)$ blocking pairs in expectation and with
	high probability. \\
\end{theorem}
\begin{proof}
	To prove the bound on the number of blocking pairs, call an evolution event on
	$y$'s list \emph{critical} if it involves the \emph{then} match of $y$. Suppose
	that after the algorithm terminates, $y\in B$ is involved in $k$ blocking
	pairs $(x_1,y),\dots,(x_k,y)$.  We observe that each one of the
	$x_1,\dots,x_k$ was involved in at least one critical evolution step. To see
	this note that if $(x,y)$ is blocking, then $x$ proposed to $y$ during the
	execution of the algorithm and got rejected subsequently (because $\pi_x$
	didn't change and $y \prec_{\pi_x} M(x)$). But since at the end of the
	execution it forms a blocking pair, it must ranked higher than $M(y)$. This is
	only possible if $x$ was swapped with the \emph{then} match of $y$ in
	some evolution event during the execution of the algorithm.

	Given this observation, we estimate the number of blocking pairs by estimating
	the number of critical evolution steps.  Note that an evolution step is
	critical with probability at most $ 2 \alpha/(n-1)$ (at most $2n$ out of
	$n(n-1)$ pairs involve the matching).  Let $T$ be a random variable equal to
	the number of proposals before the algorithm outputs a matching and label the
	corresponding time-steps as $1,2,\dots,T$.  For each step $k$, let $X_k$ be a
	Boolean random variable that is equal to 1 if at the time-step labeled $k$
	some evolution event was critical.

  First note that from a coupon-collecting argument as in
  Theorem~\ref{thm:proposalbound}
  it follows that $T=O(n\log n)$ in expectation and with high probability.
  This is because for that argument the distribution of $\{\pi_y:y\in B\}$
  is irrelevant and $\{\pi_x:x\in A\}$ being random permutations suffices.
  Therefore, we may fix an appropriately large constant $C$ so that $T\le Cn\log n$
  with high probability and let $m=Cn\log n$. 
  As noted above, at any given step and given any information from the
  previous steps, an evolution step is critical with probability at most
  $2\alpha/(n-1)$; thus, by Theorem~\ref{thm:extended-chernoff},
  \[\pr\Bigl[\sum_{k=1}^mX_k>2C\log n\Bigr]=O(n^{-C}),\]
  for sufficiently large $C$. We have
  \begin{align*}
    \pr\Bigl[\sum_{k=1}^TX_k>2C\log n\Bigr]
			&\le\pr\bigl[T>Cn\log n\bigr]
				+\pr\Bigl[\Bigl(\sum_{k=1}^TX_k>2C\log n\Bigl)\wedge\bigl(T\le Cn\log n\bigr)\Bigr]\\
     &\le\pr\bigl[T>Cn\log n\bigr]+\pr\Bigl[\sum_{k=1}^mX_k>2C\log n\Big].
  \end{align*}
  By the observation at the beginning of this paragraph the claimed bound
  holds with high probability.
  The bound on the expectation follows by noting that there can be
  at most $n^2$ blocking pairs.

  Finally, note that there will be at most $O(n\log n)$ time-steps before a
  new matching is computed. As in the final paragraph in the proof of
  Theorem~\ref{thm.matchsimple}, one can show that evolution cannot produce
  more than $O(\log n)$ blocking pairs in these many steps.
\end{proof}

\section{General Case: Improved algorithm} 

We now consider the general setting where the preference lists on both sides
may be evolving. We present a modified version of the Gale-Shapley algorithm
that takes advantage of Lemma~\ref{lem.logn} (Lemma~6 of \cite{akmu2011}) and
maintains a stable matching with at most $O((\log n)^2)$ blocking pairs.  The
algorithm consists of two separate processes that run in an interleaved
fashion: the \emph{sorting process} on even time-steps and the \emph{matching
process} on odd ones. The sorting process is basically a call to
$\textsc{SequentialSort}(\{\pi_x ~|~ x \in A\})$ that produces
\emph{approximately sorted} preference lists on the $A$ side,
$\{\tilde{\pi}_x~|~ x \in A\}$. The algorithm runs in perpetuity, in the sense
that as soon as it terminates it restarts, though the copies $\{\tilde\pi_x ~|~
x \in A\}$ from the previous execution are retained to be used by the stable
matching process. Initially, the $\tilde\pi_x$ are set to be random
permutations, thus, for the first $O(n^2 \log n)$ steps, until one run of the
sorting process is complete, the matching output by the algorithm will be
garbage. 

\begin{algorithm}
  \caption{\label{alg.match}:~~Interleaving Sorting and Matching \strut}
  \begin{algorithmic}[1]
    \For{$t=1,2,\dots$}
		\If{$t$ is \textsc{Even}} 
      	\State Perform query for Algorithm~\ref{alg.sequential}
	  	\ElsIf{$t$ is \textsc{Odd}} 
        	\State Perform query for Algorithm~\ref{alg.dyngs}
      \EndIf 
    \EndFor
  \end{algorithmic}
\end{algorithm}

We note that what is counted here is only time-steps; making queries
is the bottleneck. Additional computations required by the algorithm can be
performed in between time-steps at no additional cost.

The sorting process performs queries only during even steps and its purpose it
to keep the preference list of each $x\in A$ approximately sorted, where by
approximately sorted we meant that the conclusion of Lemma~\ref{lem.logn}
holds. 

The matching process performs queries during odd steps. Our stable matching
algorithm, which is a variant of the Gale-Shapley algorithm, is presented in
Algorithm~\ref{alg.dyngs}. Note that the $\{\tilde\pi_x ~|~ x \in A\}$ used in
Algorithm~\ref{alg.dyngs} are \emph{static} and are the output of the latest
completed run of Algorithm~\ref{alg.sequential}. However, the comparisons made
are all on dynamic data. The difference from the standard Gale-Shapley
algorithm is that whenever some $x \in A$ is about to make a proposal, first
the \emph{best} $y \in B$ among the $O(\log n)$ highest ranked as per the
ranking $\tilde\pi_x$ that have not yet rejected $x$ is found.  Note however,
that the best is with respect to the dynamic (current) preference list $\pi_x$
(otherwise, it would be trivial since $\tilde\pi_x$ is static).  This operation
is basically the algorithm to find the minimum element, which can be
implemented in $O(\log n)$ time using only comparison queries (see Section 3
in \cite{akmu2011}). We don't need to use any particular result regarding
finding the minimum element; instead, we incorporate the errors that may have
occurred while finding the minimum due to the dynamic nature of the data, into
our stable matching analysis directly.

\begin{algorithm}
  \caption{\label{alg.dyngs}:~~Modified Gale-Shapley Algorithm\strut}
  \begin{algorithmic}[1]
    \State $M \leftarrow \emptyset$
    \For{$x \in A$}
			\State $\texttt{new\_match} \leftarrow \textsc{False}$
			\State $p \leftarrow x$
			\While{$\texttt{new\_match} = \textsc{False}$}
				\State {$S \leftarrow C\log n$ \texttt{highest-ranked, not-yet-proposed
				elements in} $B$ \texttt{per} $\tilde\pi_x$}
				\State $y \leftarrow \texttt{best}(S)$ \Comment{Best with respect to
				dynamic $\pi_x$}
				\If{$M(y)$ \texttt{not yet set}}
					\State $M(p) \leftarrow y$
					\State $M(y) \leftarrow p$
					\State $\texttt{new\_match} \leftarrow \textsc{True}$
				\ElsIf{$p \prec_{\pi_y} M(y)$}
					\State $p^\prime \leftarrow M(y)$
					\State $M(y) \leftarrow p$
					\State $M(p) \leftarrow y$
					\State $p \leftarrow p^\prime$
				\EndIf
			\EndWhile
    \EndFor
    \State\Return $M$
  \end{algorithmic}
\end{algorithm}

We first describe the high-level idea of the proof. The sorting process needs
$O(n^2\log n)$ comparisons with high probability.  The approximations
$\{\tpi_x:x\in A\}$ of $\{\pi_x:x\in A\}$ that are being computed are used by
the modified Gale-Shapley algorithm. By Lemma~\ref{lem.logn} we are able
to claim that for any element $u$ in the preference list of any $x\in A$, the
number of pairs $(u,v)$ that are ordered differently in $\tpi_x$ and $\pi_x$ are
$O(\log n)$. Therefore, when $x$ is about to propose it suffices to look among
$O(\log n)$ elements in $\tpi_x$ to find the $y$ to which the proposal will be
made. Since the matching process is expected to make $O(n\log n)$ proposals, it
is expected to require $O(n(\log n)^2)$ comparisons. It turns out that, during
these steps, evolution creates a blocking pair with probability at most
$\alpha/n$. Therefore we expect $O((\log n)^2)$ blocking pairs.

We now provide the details of the proof. In order to bound the number of
blocking pairs, it is crucial that during the matching process not too many
queries are made, or alternatively that not too many proposals are made. We
therefore need an analog of Theorem~\ref{thm:proposalbound}. To apply the
coupon-collecting argument from the proof of that theorem we prove the following
lemma. 

\begin{lemma}
	Provided $\pi_x$ was chosen uniformly at random at time-step $0$ and only
	comparison queries are made, the element $y$ chosen at line 7 of
	Algorithm~\ref{alg.dyngs} is a random element from the subset of $B$ to
	which $x$ has not by that point made any proposals.
\end{lemma}
\begin{proof}
	The proof of the lemma relies on the fact that the dynamic quicksort
	algorithm used to obtain $\tilde\pi_x$ for $x \in A$ and the procedure used
	to find the best element in line 7 of Algorithm~\ref{alg.dyngs} only use
	comparison queries.

	Let $\pi_x$ be the preference list of $x$ before the first comparison is
	queried. Fix an arbitrary sequence of evolution steps that occurs during the
	computation of $y$. Suppose that given these choices of nature,
	$y=\pi_x(k)$. Then, given the same evolution steps, for any other
	permutation $\pi'_x$, $y=\pi'_x(k)$.  Since $\pi_x$ is a random permutation,
	$y=\pi_x(k)$ is a random element of $\sigma_x$.
\end{proof}

\begin{remark}
	We remark that the requirement on the implementation of \dynamicqs\ and
	\texttt{best} using only comparison queries is necessary and the lemma does
	not hold for an arbitrary algorithm.
\end{remark}

\begin{lemma}\label{lem.proposals}
  The number of proposals during one execution of the matching process is
  $O(n\log n)$ in expectation and with high probability.
\end{lemma}
\begin{proof}
  Suppose $x$ proposes to $y$ and at that point $k$ elements of $B$ are
  unmatched. Note that it must be the case that $x$ has not proposed to any of
  these elements (otherwise, they would not be currently unmatched). Thus, by
  the previous lemma, each of these $k$ elements receives a proposal with
  probability at least $1/n$. The stated bound follows from the analysis of
  coupon collector's problem as in Theorem~\ref{thm:proposalbound}.
\end{proof}

As in the one-sided setting, the analysis will rely on estimating the
occurrence of a specific kind of critical evolution steps. In the present
case the definition of a critical evolution step is a little more involved
than its one-sided counterpart.

\begin{definition}
	\label{defn:critical}
	An evolution event on the preference list $\pi_z$ of $z \in A\cup B$ is
	\emph{critical} if one of the following holds:
  	\begin{enumerate}
		\item An evolution event involves a swap of the \emph{then} match of $z$, $M(z)$.
		\item If $z \in A$, the evolution event involves swapping the
			\emph{then} best element as per $\pi_z$ to which $z$ has not yet
			proposed. 
	\end{enumerate}
\end{definition}
	  
The following claim establishes the link between the critical evolution steps
and the number of blocking pairs.

\begin{claim}
	\label{claim:critical}
	Assume that for the duration of one run of the matching process, the
	preference lists $\{\pi_z ~|~ z \in A \cup B \}$ satisfy the conditions of
	Lemma~\ref{lem.logn}, and suppose that $(x, y)$ is a blocking pair with
	respect to the returned matching. Then there was a critical evolution event
	on the preference list of at least one of $x$ and $y$ during the execution
	of the matching process.
\end{claim}
\begin{proof}
  First consider the case that $x$ proposed to $y$ during the execution of the
  matching process. It follows that $y$ rejected $x$ at some point in favor of
  some other element. When $x$ was rejected, the \emph{then} $M(y)$
  satisfies $M(y) \prec_{\pi_y} x$. Subsequently, $M(y)$ may change but has to
  become better, unless there was a swap that involves the \emph{then} $M(y)$,
  which is a critical event on $\pi_y$. Since, we know that in the end $x
  \prec_{\pi_y} x^\prime$, where $x^\prime$ is the final match of $y$, there
  must have been some evolution event where $x$ was swapped with the
  \emph{then} match of $y$. Thus, by part 1 of Definition~\ref{defn:critical},
  a critical evolution event occurred on $\pi_y$.

	On the other hand, suppose $x$ never proposed to $y$, and let $y^\prime$ be
	the final match of $x$.  Suppose that when $x$ proposed to $y'$, $y
	\prec_{\pi_x} y'$; it follows that $\texttt{best}$ on line~7 of
	Algorithm~\ref{alg.dyngs} failed to return the best element to which $x$ had
	not yet proposed. Since, we are assuming that $\tilde\pi_x$ is a sufficient
	approximation of $\pi_x$, it must be because the actual \emph{best} element
	was swapped at least once while $\texttt{best}$ was being executed. Thus, by
	part 2 of Definition~\ref{defn:critical}, a critical evolution event
	occurred on $\pi_x$. Finally, if when $x$ proposed to $y^\prime$, it was the
	case that $y^\prime \prec_{\pi_x} y$, but $(x, y)$ is a blocking pair, it
	must be that $y^\prime$ was involved in a swap subsequently leading to a
	critical event involving $x$'s the \emph{then} match.
\end{proof}

\begin{theorem} \label{thm.improved}
	Provided the initial preference lists are drawn randomly,\footnote{This is
	not actually required, since after sufficiently long (though still
	polynomial) time, all the preference lists will be close to random due to a
	mixing time argument on the set of permutations.} for all $z\in A\cup B$, for a sufficiently large constant $c_0$ and any time-step $t \geq c_0 n^2 \log n$, 
	Algorithm~\ref{alg.dyngs} maintains a matching with at most $O((\log n)^2)$
	blocking pairs in expectation and with high probability. \\
\end{theorem}
\begin{proof}
	As a result of Claim~\ref{claim:critical}, we can estimate the number of
	blocking pairs by estimating the number of critical evolution steps.  Let
	$T$ be a random variable equal to the number of queries of Algorithm~\ref{alg.dyngs}
	before it outputs a matching and label the corresponding 
	time-steps as $1,2,\dots,T$.  For each step $k$, let $X_k$ be a Boolean
	random variable that is equal to 1 if at the time-step labeled $k$ the
	evolution was critical.  Furthermore, denote by ${\mathcal E}$ the event
	that during these $T$ time-steps the lists were approximately sorted. By
	Lemma~\ref{lem.logn}, the event ${\mathcal E}$ occurs with high probability.

	By Lemma~\ref{lem.proposals} it follows that $T=O(n(\log n)^2)$ in expectation
  and with high probability, since we are wasting $O(\log n)$ queries per
  proposal.  Therefore, we may fix an appropriately large constant $C$ so that
  $T\le Cn(\log n)^2$ with high probability and let $m=Cn(\log n)^2$.  Note
  that---given any history of evolution steps---an evolution step is critical
  with probability at most $O(\alpha/n)$, since for each $z\in A\cup B$ there
  is a constant number of elements that evolution has to swap in order to be
  critical.
  Thus, by Theorem~\ref{thm:extended-chernoff},
  \[\pr\Bigl[\sum_{k=1}^mX_k>2C(\log n)^2\Bigr]=O(n^{-C}),\]
  for sufficiently large $C$. We have
  \begin{align*}
    \pr\Bigl[\sum_{k=1}^TX_k>2C(\log n)^2\Bigr]
      &\le\pr\bigl[T>Cn(\log n)^2\bigr]+\pr[\bar{\mathcal E}]\\
      &\qquad+\pr\Bigl[\Bigl(\sum_{k=1}^TX_k>2C(\log n)^2\Bigl)\wedge
        \bigl(T\le Cn(\log n)^2\bigr)\wedge{\mathcal E}\Bigr]\\
     &\le\pr\bigl[T>Cn(\log n)^2\bigr]+\pr[\bar{\mathcal E}]
      +\pr\Bigl[\sum_{k=1}^mX_k>2C(\log n)^2\Big].
  \end{align*}
  It follows that the claimed bound
  holds with high probability.
  The bound on the expectation follows by noting that there can be
  at most $n^2$ blocking pairs.
\end{proof}

\par

\section*{Acknowledgments}

For earlier discussions, F.M.\ would like to thank Marcos Kiwi who, among other
things, introduced him to the line of works on evolving data and shared
preliminary thoughts on the possibility of computing stable matchings in this
context. We also thank a reviewer for (minor) corrections in our main theorem.

\bibliographystyle{alpha}
\bibliography{esm}	
\end{document}